\documentclass[a4paper,11pt]{article}
\usepackage[english]{babel}
\usepackage[intlimits]{amsmath}
\usepackage{amssymb,amsthm,graphicx,latexsym,bbm,calc}
\usepackage[a4paper]{geometry}
\usepackage[colorlinks,final,plainpages=false,pdfpagelabels]{hyperref}

\newcommand\kreuz{\Sigma}
\newcommand{\defin}{\mathrel{:=}}
\newcommand{\indef}{\mathrel{=:}}

\newcommand{\mathd}{\mathrm{d}}
\newenvironment{itemizedot}{\begin{itemize} }{\end{itemize}}
\theoremstyle{plain}
\newtheorem{proposition}{Proposition}
\newtheorem{corollary}[proposition]{Corollary}
\theoremstyle{definition}
\newtheorem{definition}[proposition]{Definition}
\newtheorem{remark}{Remark}

\begin{document}

\title{A mathematical framework for raw counts\\
  of single-cell RNA-seq data analysis}


\author{Silvia~Giulia~Galfr{\`e}\thanks{University of Roma Tor Vergata
    and Scuola Normale Superiore, Pisa.} \and
  Francesco~Morandin\thanks{Department of Mathematical, Physical and Computer Sciences, University of Parma.}}

\maketitle

\begin{abstract}
  Single-cell RNA-seq data are challenging because of the sparseness
  of the read counts, the tiny expression of many relevant genes, and
  the variability in the efficiency of RNA extraction for different
  cells. We consider a simple probabilistic model for read counts,
  based on a negative binomial distribution for each gene, modified by
  a cell-dependent coefficient interpreted as an extraction
  efficiency. We provide two alternative fast methods to estimate the
  model parameters, together with the probability that a cell results
  in zero read counts for a gene. This allows to measure genes
  co-expression and differential expression in a novel way.
\end{abstract}

\section{Introduction}

In recent years the availability of rich biological datasets is
challenging the flexibility and robustness of statistical
techniques. In fact, while when the size of the sample is moderate it
is customary and accepted to make quite strong assumptions on the
underlying distributions, in the contest of big data this could often
lead to obvious distortions and inconsistencies. This can be relevant
in particular in the case of ``omics'' data (proteomics, genomics or
transcriptomics), of which single-cell RNA sequencing (scRNA-seq) is a
very recent and exceptionally difficult
example~\cite{scRNAseq-Rev2015,RosSve2017,LueThe2019}.

Single-cell RNA-seq data are large matrices with genes in the rows and
single cells in the columns, with integer read counts in each
component.  The vast majority (80\% and more) of the read counts are
zero and an even larger fraction (95\% and more) of the genes has an
average of less than 1 read count per cell. Nonetheless it is believed
that around 20\% of genes are typically active and functional in a
cell, and many of these are transcription factors, whose subtle
modulation controls the functions and state of the cell.

Given this preamble it is not surprising that the analysis of
scRNA-seq is a very important topic and that a general robust approach
is yet to be found. In this paper we present a new promising
mathematical framework, and propose suitable parameter estimators and
statistical inference for gene co-expression.

Most statistical models for scRNA-seq data, deal with what is usually
called \emph{level of expression}, which is obtained from read counts
by pseudocount addition and log-transformation
(see~\cite{Chen-scRNAseq-Rev2019,Choi-scRNAseq-Rev2019} and references
therein). There have been many attempts to normalize and
variance-stabilize these quantities, but it remains a difficult
problem (see for example~\cite{RosSve2017,SAVER2018,HafSat2019}).

Our probabilistic model, on the other hand, belongs to the part of the
literature that deals with the raw integer read counts (see among the
others~\cite{Deseq2-2014,SCDE2014,basics2015,scImpute2018,SAVER2018,AmrHar2019}
and references therein). Our model is in particular similar to the one
introduced by BASiCS~\cite{basics2015}, but we use it in a
non-Bayesian contest and we do not model spike-ins.

For each cell $c$ and gene $g$, we model the number of read counts
$R_{g, c}$ with conditional Poisson distribution
\[
R_{g,c}|\nu_c,\Lambda_g^{(c)}\sim\operatorname{Poisson}(\nu_c\Lambda_g^{(c)})
\]
depending on:
\begin{itemizedot}
\item a deterministic {\emph{extraction efficiency parameter}} $\nu_c$
  which modulates the expression of all genes for that cell, and
\item a random \emph{potential biological expression level}
  $\Lambda_g^{(c)}$ for each gene.
\end{itemizedot}
For the first part of the paper, when estimating $\nu_c$ and
$\lambda_g\defin E(\Lambda_g^{(c)})$,  we make no assumptions on the
distribution $\mathcal L_g$ of $\Lambda_g^{(c)}$

In the second part, we need to estimate the probability of zero read
counts $P(R_{g,c}=0)$, and to this end we make the further assumption
that this probability can be approximated by assuming that
$\mathcal L_g$ is gamma with mean $\lambda_g$, and
variance $a_g\lambda_g^2$ that can be fitted on the total number of
zero read counts for that gene.  Equivalently, $R_{g,c}$ is considered
of negative binomial distribution with mean $\nu_c\lambda_g$ and
dispersion $a_g$.

We remark that this assumption does not concern the whole distribution
of the read counts, but only the probability of zero, which then takes
the form typical of the negative binomial,
\[
  P(R_{g,c}=0)
  \approx\biggl(\frac{a_g^{-1}}{\nu_c\lambda_g+a_g^{-1}}\biggr)^{a_g^{-1}}.
\]

Often in the literature the frequency of zero read counts has been
considered not completely explained, when using the most natural
statistical models, and the concept of \emph{dropout} has been
introduced~\cite{SCDE2014,ZIFA2015,MAST2015,SCDD2016}. Recently there
have been criticism on this subject~\cite{Sve2020} and it is unclear
if the need of zero-inflated distributions is really a technical
issue, or in fact it is an artifact due to the use of log-transformations,
or limited to the case of non-UMI datasets~\cite{VieEtA2017}.

In our model, zero read counts are considered effects of biological
variability and random extraction, and in fact the inference itself is
based precisely on the occurrence of these events.

After the model is introduced in Section~\ref{s:model}, the remainder
of the paper is organized as follows.

In Section~\ref{s:estimation} we propose two fast methods to get
estimates of the relevant parameters, both based on moment estimation,
and discuss their validity. Maximum likelihood estimation is in good
accordance with our methods, but it requires more resources and has to
assume the class of $\mathcal L_g$ (typically gamma); this is a choice
that we defer until the inference in subsequent sections.

In Section~\ref{s:probzero} we introduce a way to estimate the
probability of zero read counts, using the estimated parameters and
making some assumptions on $\mathcal L_g$, in particular that it can
be approximated by a gamma distribution. A natural way to estimate its
second parameter, is to fit the total number of cells with zero read
counts for gene $g$.

In Section~\ref{s:tables} we build co-expression tables, which are similar to
contingency tables, but count the number of cells in which two genes have been
found expressed together. It is shown that these cannot be analysed like
classical contingency tables, because the different efficiency of the cells
would cause spurious correlations. Nevertheless the estimates built on the
previous sections allow to design a statistical test for independence and a
co-expression index. Extensions to differential expression analysis and to a
global differentiation index are discussed.

In Section~\ref{s:synthetic} we report the results of the numerical
simulations with synthetic datasets, used to evaluate the estimators, the
distribution of the statistics, and the false positive rate of the tests.

A twin paper with a computational-biology point of view (currently in the
final stages of processing), deals with the application of this framework to
real biological datasets and includes the software implementation of all the
tools.

\section{Model}

\label{s:model}Single-cell RNA-seq data analysis is generally
performed on a huge matrix of counts
$R = (R_{g, c})_{g \in G, c \in C}$, where $G$ and $C$ are the sets of
genes and cells respectively. Typical sizes are
$n \defin | G | \sim 15000$ and $m \defin | C | \sim 1000$--$10000$. The
read counts $R_{g, c}$ are non-negative integers, with many zeros.

Usually for bulk RNA-seq, where there is no information at single cell level,
the counts $R_g$ are modeled with the gamma-Poisson mixture (also known as
negative binomial distribution), which is quite suited to the need, as it is
supported on the non-negative integers and has two real parameters that ensure
a good flexibility (see~\cite{edgeR2010,Deseq2-2014} among the others).

From a physical point of view, this can be interpreted as a model in which the
total amount of RNA molecules of gene $g$ is approximated by a gamma random
variable $\Lambda_g \sim \text{gamma} (\eta_g, \theta_g)$ with parameters
depending on $g$, and the number of reads has then Poisson conditional
distribution $R_g | \Lambda_g \sim \text{Poisson} (\nu \Lambda_g)$ with a
small efficiency $\nu$.

One of the challenges of single-cell RNA-seq is that one should consider a
different efficiency $\nu_c$ for each cell $c$, and that a single gamma
distribution may not be able to account for two or more cell conditions or
types inside the experiment's population.

To reduce technical noise, which in our model is not accounted for, we
make the assumption of dealing with a post-quality-control scRNA-seq
dataset with UMI\footnote{Unique Molecular Identifiers are molecular
  labels that nearly eliminate amplification noise~\cite{umi2013}.}
counts as input.

Given these assumptions, we will model the counts $R_{g, c}$ as random
variables with Poisson conditional distribution
\begin{equation}
  R_{g, c} | \Lambda_g^{(c)} \sim \text{Poisson} (\nu_c \Lambda_g^{(c)}),
  \qquad \text{(conditionally independent)}
\end{equation}
and the real number of molecules $\Lambda_g^{(c)}$ with some unknown
distribution.

Since $\nu_c$ and $\Lambda_g^{(c)}$ are everywhere multiplied together, they
can only be known up to a multiplicative constant. Without loss of generality,
we will assume throughout this paper that this constant is fixed in such a way
that
\begin{equation}
  \label{e:nustar} \nu_{\ast} \defin \frac{1}{m} \sum_{c \in C} \nu_c = 1,
\end{equation}
hence $\Lambda_g^{(c)}$ will be rescaled accordingly, and it will not represent
the real number of molecules, but just some typical value for the counts.

We will suppose that, for $c \in C$, the columns $\Lambda^{(c)} \defin
(\Lambda_g^{(c)})_{g \in G}$ are i.i.d.~random vectors with distribution
$\mathcal{L}$ on $\mathbb{R}_+^G$, and that $\mathcal{L}$ has expectation
$\lambda = (\lambda_g)_{g \in G}$ and covariance matrix $Q \defin (Q_{g,
h})_{g, h \in G}$ so that,
\[ \lambda_g \defin E (\Lambda_g^{(c)}) \qquad
   \text{and} \qquad Q_{g, h} \defin \operatorname{Cov}
   (\Lambda_g^{(c)}, \Lambda_h^{(c)}), \qquad c \in C
\]
In Section~\ref{s:estimation} we will show how to estimate the parameters
$(\nu_c)_{c \in C}$ and $(\lambda_g)_{g \in G}$. The biological information on
the differentiation of the cells in the sample, is instead encoded inside $Q$
and will be the subject of the subsequent sections.

\section{Parameter estimation}

\label{s:estimation}A direct computation shows that
\begin{equation}
  \label{e:mu-gc} \mu_{g, c} \defin E (R_{g, c}) = E [E (R_{g, c} |
  \Lambda_g^{(c)})] = \nu_c E (\Lambda_g^{(c)}) = \nu_c \lambda_g .
\end{equation}
The quantity $\mu_{g, c}$ represents the expected read count number, and takes
into account the efficiency $\nu_c$ of cell $c$ and the average expression
level $\lambda_g$ of gene $g$.

The formula for the variance can be obtained similarly, but it depends on one
additional parameter $a_g \defin \frac{\operatorname{Var} (\Lambda_g^{(c)})}{E
(\Lambda_g^{(c)})^2}$ and will not be used much, but we give it for
completeness and reference,
\begin{equation}
  \label{e:Var-R} \operatorname{Var} (R_{g, c}) = E [\operatorname{Var} (R_{g, c} |
  \Lambda_g^{(c)})] + \operatorname{Var} [E (R_{g, c} | \Lambda_g^{(c)})] = \mu_{g, c}
  + a_g \mu_{g, c}^2 .
\end{equation}
Notice that the non-homogeneous dependence on $\nu_c$ explains quite
well the fact that no scaling factor can be used to normalize data so
that the variance is stabilized~\cite{HafSat2019}.

In the remainder of this section we develop two fast methods to
estimate $\mu_{g, c}$ for all genes $g$ and cells $c$. The first one
is simple and straightforward, but may sometimes be affected by few
genes with high level of expression and large biological
variability. The second one is based on a variance stabilizing
transformation that, to our knowledge, is used here for the first time
for scRNA-seq data analysis. It shows some small bias but should be
more stable with respect to random variations in the most expressed
genes.

Both these methods are based on moment estimation and do not assume
anything about the distribution $\mathcal L_g$. Maximum likelihood
estimation on the other hand may be preferred when the distribution of
$\mathcal L_g$ can be safely assumed to be gamma. For example this is
the case of a single cluster of cells of similar expression, and we
used this approach in Section~\ref{s:synthetic} to estimate
parameters to generate synthetic datasets. We do not delve into this
matter here.

Even though we give some provable statements to establish good
properties of our estimators, it is quite difficult to assess their
precision and accuracy. Section~\ref{s:synthetic} explains how we
generated several \emph{realistic} synthetic datasets and used them to
gauge the estimators. Figure~\ref{f:nu_lambda_est} shows the results.

\subsection{Average estimation}

The most natural way to estimate the parameters is the following. Define the
rows, columns and global averages by
\begin{equation}
  R_{g, \ast} \defin \frac{1}{m} \sum_{c \in C} R_{g, c}, \qquad R_{\ast, c} \defin \frac{1}{n} \sum_{g \in G} R_{g, c},
  \qquad R_{\ast, \ast} \defin \frac{1}{mn} \sum_{g,
  c} R_{g, c} .
\end{equation}
\begin{definition}
  The {\emph{average estimators}} of the parameters, marked with the ``hat''
  symbol, are given by
\[ \hat{\lambda}_g \defin R_{g, \ast}, \qquad
   \hat{\nu}_c \defin \frac{R_{\ast, c}}{R_{\ast, \ast}}, \qquad \text{and} \qquad \hat{\mu}_{g, c}
   \defin \frac{R_{g, \ast} \cdot R_{\ast, c}}{R_{\ast, \ast}} . \]
\end{definition}
\begin{proposition}
  The average estimator of $\lambda_g$ is unbiased.
  Moreover $E (R_{\ast,
  \ast}) = \lambda_{\ast}$ and $E (R_{\ast, c}) = \nu_c \lambda_{\ast}$.
\end{proposition}

\begin{proof}
  By equations~{\eqref{e:nustar}} and~{\eqref{e:mu-gc}},
  \[ E (\hat{\lambda}_g) = \frac{1}{m} \sum_{c \in C} E (R_{g, c}) =
     \frac{1}{m} \sum_{c \in C} \nu_c \lambda_g = \lambda_g \]
  and analogously for the other cases.
\end{proof}

On some real biological datasets this appears to be a poor way to
estimate the unknown parameters. In particular there is evidence that
$R_{\ast, c}$ may be too sensible to the few genes that have both high
reads and large biological variability between cells. Since we plan to
use estimates of $\nu_c$ to normalize the dataset, this would be a
source of spurious correlations, and difficult to deal with.

\subsubsection{Problems of average estimation}

A mathematical explaination of the occasional weakness of these estimators
could be the following.

Suppose we are looking for weights $(w_g)_{g \in G}$ such that a linear
combination of the counts $A_c (w) \defin \sum_{g \in G} w_g R_{g, c}$ is a
good estimator of $\nu_c$. Notice that $\hat{\nu}_c$ is such an estimator, and
it is characterized by having uniform weights $w_g \defin \bar{w}$, whose
value is fixed by the additional constraint that $\frac{1}{m} \sum_{c \in C}
A_c (w) = 1$,
\[ 1 = \frac{1}{m} \sum_{c \in C} A_c (w) = \frac{1}{m} \sum_{c \in C} \sum_{g
   \in G} w_g R_{g, c} = \text{} \sum_{g \in G} w_g  \hat{\lambda}_g \]
yielding $\bar{w} = n^{- 1} R_{\ast, \ast}^{- 1}$.

With this insight, let us consider $A_c (w)$ under the somewhat simpler
constraint $\sum_g w_g \lambda_g = 1$. Then $E (A_c) = \nu_c$, so $A_c (w)$ is
an unbiased estimator of $\nu_c$ for all choices of the weights. Since there
is no independence between counts of different genes, the variance is more
complicated and must be computed with conditional expectations,
\[ \operatorname{Var} (A_c) = E [\operatorname{Var} (A_c | \Lambda^{(c)})] + \operatorname{Var} (E
   [A_c | \Lambda^{(c)}]) = \nu_c  \sum_g w_g^2 \lambda_g + \nu_c^2  \langle
   w, Qw \rangle . \]
If the term $\nu_c^2  \langle w, Qw \rangle$ was not present, the variance of
$A_c (w)$ would have been minimized, under the constraint, by choosing $w_g
\equiv \text{const}$, as a direct computation with Lagrange multipliers shows.
The presence of this term, on the other hand, hints that the weights should be
smaller for genes with large biological variability. Unfortunately it is very
difficult to estimate it, as it depends on the whole covariance matrix $Q$,
which is what actually holds the biological information on the differentiation
of the cells in the experiment's population.

Apart for this sub-optimality of the constant weights in terms of total
variance, a second problem (which may even be more serious) is that even with
optimal weights, the estimator would correlate in particular with high
variance genes, while one of our targets is to have it as much uncorrelated as
possible to single genes.

\subsection{Square root estimation}

To get estimates that may be more robust in the cases where \emph{average
estimators} are not, we recall that the square root of a Poisson random
variable of mean $x$ has variance $\tau (x)$ which depends weakly on $x$, in
particular, $\tau (x) \rightarrow 1 / 4$ as $x \rightarrow \infty$. This
useful property is at the base of a classical variance-stabilizing
transformation that is expected to improve the robustness of averages at the
cost of adding a small additional bias.

Let us introduce the \emph{square root counts} and their rows and columns
averages,
\begin{equation}
  X_{g, c} \defin \sqrt{R_{g, c}}, \qquad X_{g, \ast}
  \defin \frac{1}{m} \sum_{c \in C} X_{g, c}, \qquad
  X_{\ast, c} \defin \frac{1}{n} \sum_{g \in G} X_{g, c} .
\end{equation}
We will need also the corresponding sample variances
\begin{equation}
  S_{g, \ast}^2 \defin \frac{1}{m - 1} \sum_{c \in C} (X_{g, c} - X_{g,
  \ast})^2, \qquad S_{\ast, c}^2 \defin \frac{1}{n -
  1} \sum_{g \in G} (X_{g, c} - X_{\ast, c})^2 .
\end{equation}
Then we introduce our main estimators, whose properties will be analyzed in
the remarks and proposition below.

\begin{definition}
  The {\emph{square-root estimators}} of the parameters, marked with the
  ``check'' symbol, are given by
  \begin{align*}
    \check{\lambda}_g & \defin \psi (X_{g, \ast}) + \frac{1}{2} \psi''
    (X_{g, \ast}) \cdot \left[ \frac{m - 1}{m} S_{g, \ast}^2 - \psi (X_{g,
    \ast}) + X_{g, \ast}^2 \right], \qquad g \in G\\
    \check{\nu}_c & \defin  \frac{\tilde{\nu}_c}{\tilde{\nu}_{\ast}}
    \defin \frac{\tilde{\nu}_c}{\frac{1}{m} \sum_{u \in C} \tilde{\nu}_u},
    \qquad c \in C\\
    \check{\mu}_{g, c} & \defin  \check{\lambda}_g \check{\nu}_c,
  \end{align*}
  where
  \[ \tilde{\nu}_c \defin \psi (X_{\ast, c}) + \frac{1}{2} \psi'' (X_{\ast,
     c}) \cdot \left[ \frac{n - 1}{n} S_{\ast, c}^2 - \psi (X_{\ast, c}) +
     X_{\ast, c}^2 \right], \qquad c \in C, \]
  and $\psi = \varphi^{- 1}$ is the inverse of the function $\varphi :
  \mathbb{R}_+ \rightarrow \mathbb{R}_+$ defined by
  \[ \varphi (x) \defin E \left[ \sqrt{\text{Poisson} (x)} \right] \defin
     \sum_{k \geq 1} \sqrt{k}  \frac{x^k}{k!} e^{- x},
     \qquad x \geq 0. \]
\end{definition}

\begin{remark}
  The main term of the formula for $\check{\lambda}_g$ is $\psi (X_{g,
  \ast})$ and for large $x$, we have $\psi (x) \approx x^2$, so
  $\check{\lambda}_g \approx X_{g, \ast}^2$ plus some correction terms, and
  analogously for $\tilde{\nu}_c$. (See Figure~\ref{f:tau_psi} below.)
  
  To see that $\psi (x) \approx x^2$, let be $x \geq 0$ and consider $R
  \sim \text{Poisson} (\psi (x))$; then $E \left[ \sqrt{R} \right] = \varphi
  (\psi (x)) = x$ and $E [R] = \psi (x)$, so that
  \[ \psi (x) - x^2 = \operatorname{Var} \left( \sqrt{R} \right) = \tau (\psi (x)) \in
     [0, L], \]
  where $L = \max_x \tau (x) \approx 0.4125$. This also implies that $\varphi
  (x) = \sqrt{x - \tau (x)}$.
  
  \begin{figure}[ht]
    \includegraphics[height=52mm]{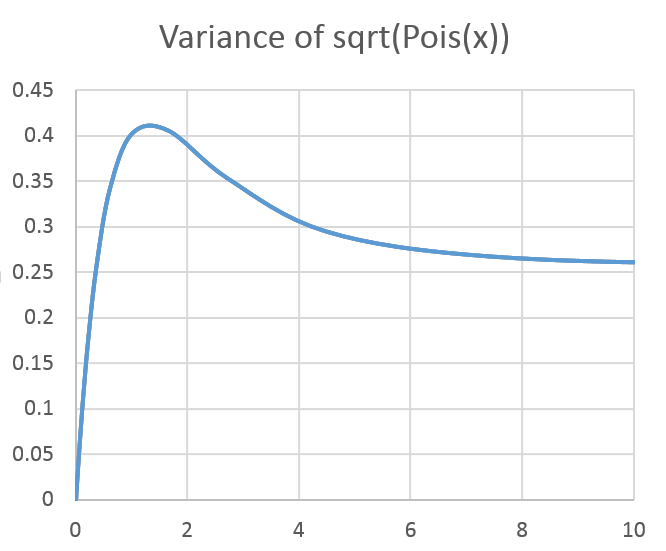}
    \makebox[60mm][l]{\includegraphics[height=54mm]{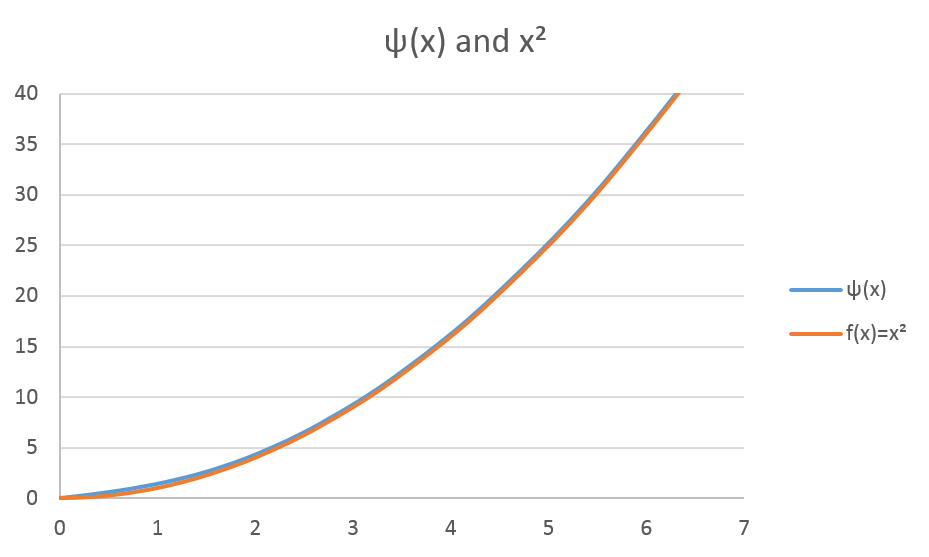}}
    \caption{\label{f:tau_psi}Plot of $\tau (x)$ and of $\psi (x)$ together
    with $x^2$.}
  \end{figure}
\end{remark}

\begin{remark}
  We stress that square (or square root) and average do not commute, and in
  particular by Jensen inequality we get,
  \[ X_{g, \ast}^2 = \left( \frac{1}{m} \sum_{c \in C} X_{g, c} \right)^2
     \leq \frac{1}{m} \sum_{c \in C} X_{g, c}^2 = \hat{\lambda}_g . \]
  Hence, as $\hat{\lambda}_g$ is an unbiased estimator, $X_{g, \ast}^2$ in
  itself would be a poor estimator of $\lambda_g$, with systematic negative
  bias. The square root estimator $\check{\lambda}_g$ is a second order
  correction of the above approach.
\end{remark}

\begin{figure}[p]
  
  \includegraphics[width=0.88\textwidth]{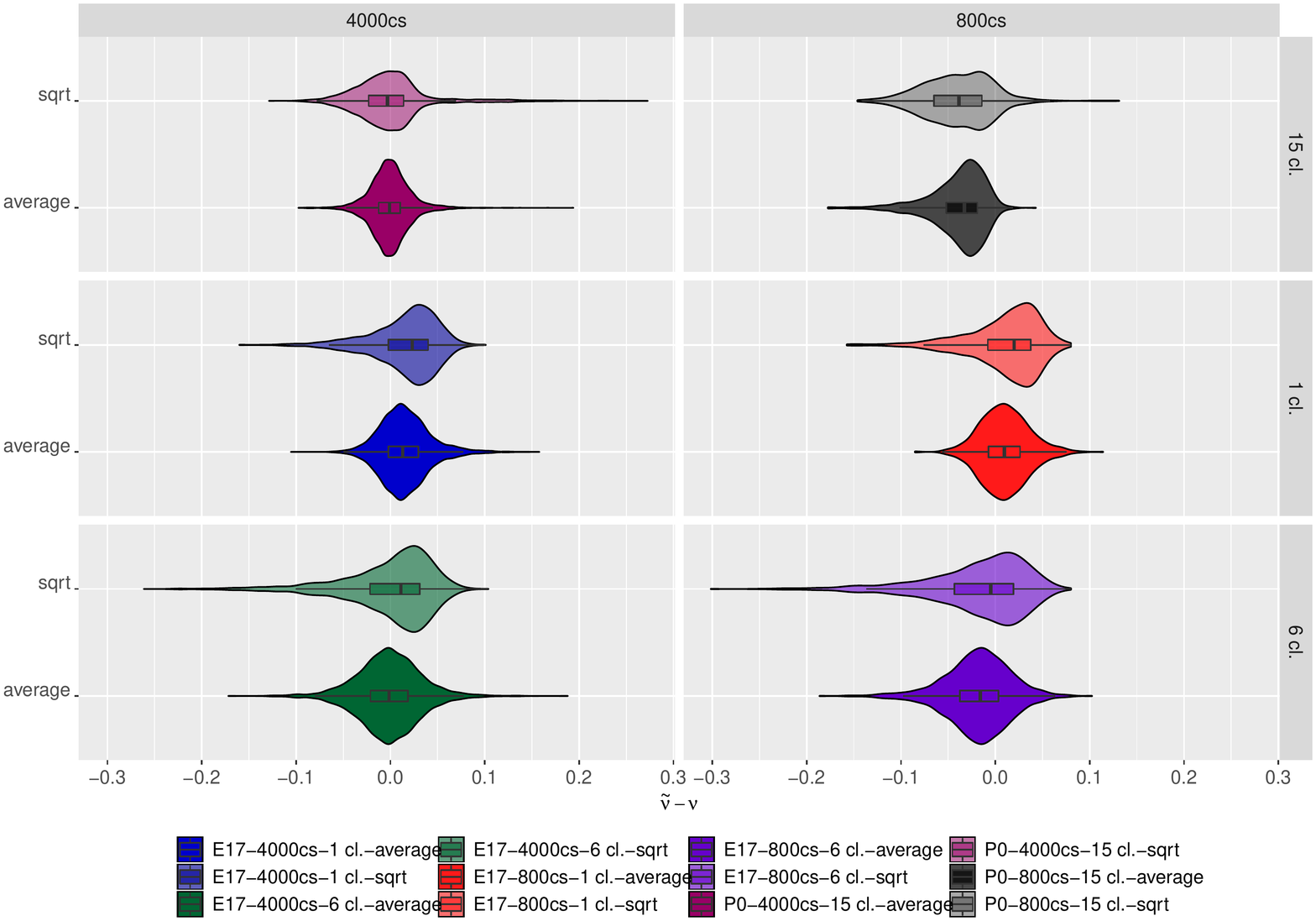}

  \medskip

  \includegraphics[width=0.88\textwidth]{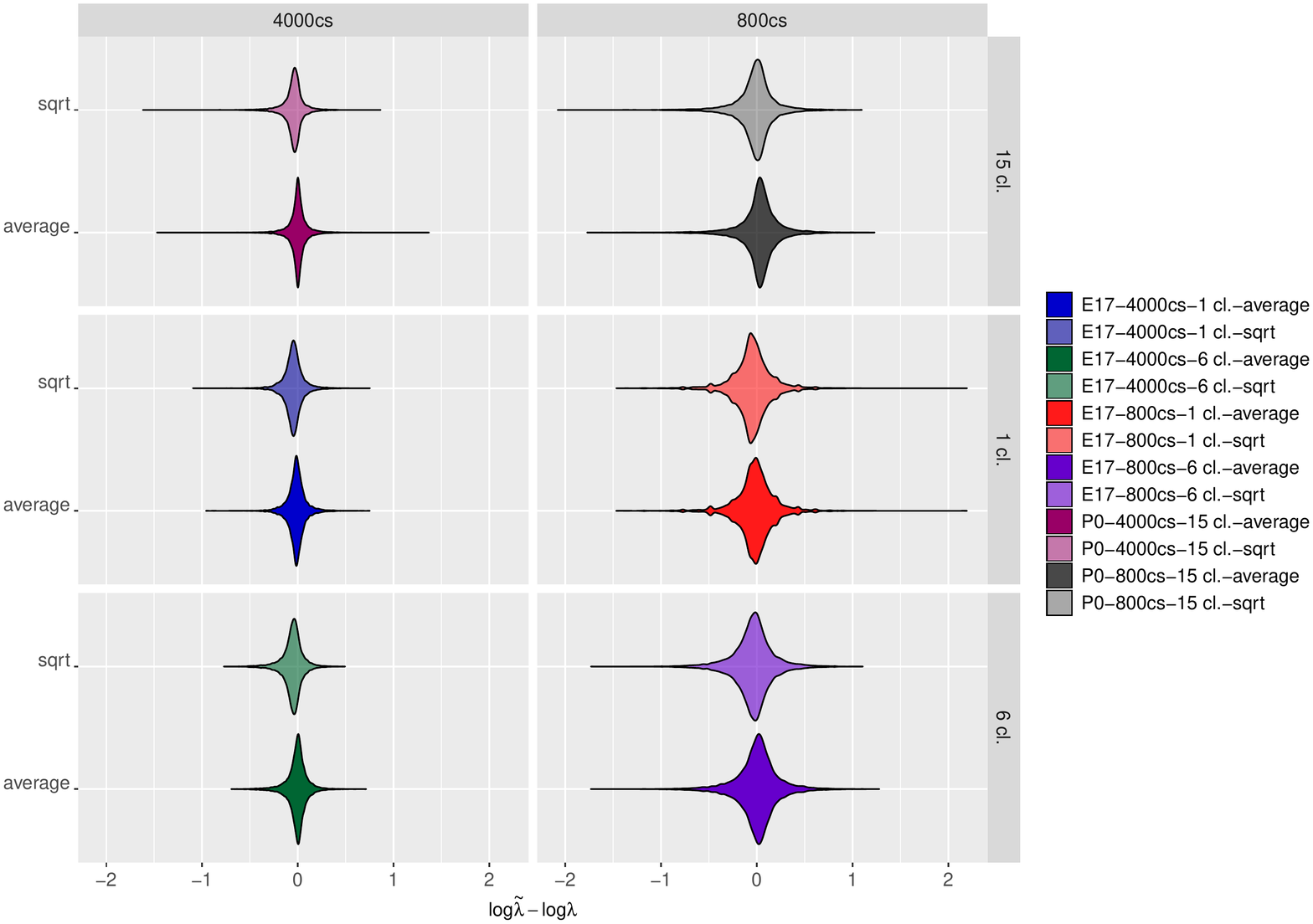}
  \caption{Accuracy and precision of $\nu$ and $\lambda$
    estimators. Recall that $\nu_\ast=1$, so the typical value of
    $\nu_c$ is 1 and therefore the typical CV of the estimators for
    $\nu$ is about $0.04$. The corresponding result for $\lambda$
    shows a strong dependence on the number of cells.}
  \label{f:nu_lambda_est}
\end{figure}

The following proposition gives a non-rigorous argument to show that the terms
in the square root estimators are the right ones to get the smallest bias.

\begin{proposition}
  The statistics $\check{\lambda}_g$, $\check{\nu}_c$ and
  $\check{\mu}_{g, c}$ estimate $\lambda_g$, $\nu_c$ and $\mu_{g, c}$ with
  a small bias depending on the unknown distribution $\mathcal{L}$ of
  $\Lambda_g^{(c)}$ and an additional error of order $m^{- 1 / 2}$.
\end{proposition}

\begin{proof}
  The first step is to approximate $\lambda_g$ with
  \[ \lambda_g = \lambda_g \nu_{\ast} \approx \frac{1}{m} \sum_{c \in C} \nu_c
     \Lambda_g^{(c)}, \]
  in fact $E (\Lambda_g^{(c)}) = \lambda_g$ and by independence the error is
  of the order $m^{- 1 / 2}$,
  \[ \left| \lambda_g - \frac{1}{m} \sum_{c \in C} \nu_c \Lambda_g^{(c)}
     \right| \lesssim \frac{1}{\sqrt{m}} . \]
  Then we write $\nu_c \Lambda_g^{(c)} = \psi (\varphi (\nu_c
  \Lambda_g^{(c)}))$ and then approximate $\psi$ with a Taylor expansion to
  the second order,
  \[ \psi (x) \approx \psi (x_0) + \psi' (x_0) \cdot (x - x_0) + \frac{1}{2}
     \psi'' (x_0) \cdot (x - x_0)^2 . \]
  If we substitute $x = \varphi (\nu_c \Lambda_g^{(c)})$ and $x_0 =
  \tilde{X}_g \defin \frac{1}{m} \sum_{c \in C} \varphi (\nu_c
  \Lambda_g^{(c)})$ and also average the whole expression for $c \in C$, then
  the linear term disappears:
  \[ \frac{1}{m} \sum_{c \in C} \nu_c \Lambda_g^{(c)} \approx \psi
     (\tilde{X}_g) + \frac{1}{2} \psi'' (\tilde{X}_g) \cdot W_g, \]
  where
  \[ W_g \defin \frac{1}{m} \sum_{c \in C} (\varphi (\nu_c \Lambda_g^{(c)}) -
     \tilde{X}_g)^2, \]
  with an error of order proportional to the third moment (skewness) of
  $\mathcal{L}$. (We remark moreover that $| \psi''' (x) | \leq K \approx
  1.206$ for all $x \geq 0$.)
  
  To use the above formulas, we will approximate $\tilde{X}_g$ with $X_{g,
  \ast}$ and $W_g$ with $S_{g, \ast}^2$, adjusted appropriately.
  
  To do so, let $\Lambda_g$ denote the vector of i.i.d.~random variables
  $(\Lambda_{g^{}}^{(c)})_{c \in C}$ and consider the following conditional
  expectations,
  \begin{align*}
    E [X_{g, c} | \Lambda_g^{}] &= \varphi (\nu_c \Lambda_g^{(c)})\\
    E [X_{g, c}^2 | \Lambda_g^{}] &= \nu_c \Lambda_g^{(c)}\\
    \operatorname{Var} [X_{g, c} | \Lambda_g^{}] &= \nu_c \Lambda_g^{(c)} - \varphi
    (\nu_c \Lambda_g^{(c)})^2 = \tau (\nu_c \Lambda_g^{(c)})\\
    E [X_{g, \ast} | \Lambda_g] &= \frac{1}{m} \sum_{c \in C} \varphi
    (\nu_c \Lambda_g^{(c)}) = \tilde{X}_g\\
    \operatorname{Var} [X_{g, \ast} | \Lambda_g] &= \frac{1}{m^2} \sum_{c \in C}
    \operatorname{Var} [X_{g, c} | \Lambda_g^{(c)}] = \frac{1}{m^2} \sum_{c \in C}
    \tau (\nu_c \Lambda_g^{(c)}) \leq \frac{L}{m} .
  \end{align*}
  Hence we get immediately that $X_{g, \ast}$ approximates $\tilde{X}_g$ with
  an error of order $m^{- 1 / 2}$,
  \[ | \tilde{X}_g - X_{g, \ast} | \lesssim \frac{1}{\sqrt{m}} . \]
  Finally we need to approximate $W_g$. We start from $S_{g, \ast}^2$:
  \begin{multline*}
    E [(X_{g, c} - X_{g, \ast})^2 | \Lambda_g]  = E [X_{g, c} - X_{g, \ast}
    | \Lambda_g]^2 + \operatorname{Var} [X_{g, c} - X_{g, \ast} | \Lambda_g]\\
    =(\varphi (\nu_c \Lambda_g^{(c)}) - \tilde{X}_g)^2 + \operatorname{Var}
    [X_{g, c} | \Lambda_g] + \frac{1}{(m - 1)^2} \sum_{c' \neq c} \operatorname{Var}
    [X_{g, c'} | \Lambda_g]\\
    =(\varphi (\nu_c \Lambda_g^{(c)}) - \tilde{X}_g)^2 + \tau (\nu_c
    \Lambda_g^{(c)}) + \frac{1}{(m - 1)^2} \sum_{c' \neq c} \tau (\nu_{c'}
    \Lambda_g^{(c')}) .
  \end{multline*}
  Averaging for $c \in C$ (with denominator $m - 1$) yields,
  \[
  \begin{split}
    E [S_{g, \ast}^2 | \Lambda_g] &= \frac{1}{m - 1} \sum E [(X_{g, c} -
    X_{g, \ast})^2 | \Lambda_g]\\
    &= \frac{1}{m - 1} \sum_{c \in C} (\varphi (\nu_c \Lambda_g^{(c)}) -
    \tilde{X}_g)^2 + \frac{1}{m - 1} \sum_{c \in C} \tau (\nu_c
    \Lambda_g^{(c)}),
  \end{split}
\]
  and hence we do the following approximation, with an error of order $m^{- 1
  / 2}$.
  \[ W_g \approx \frac{m - 1}{m} S_{g, \ast}^2 - \frac{1}{m} \sum_{c \in C}
     \tau (\nu_c \Lambda_g^{(c)}) . \]
  The last term cannot be consistently estimated, and neither can one use
  Bayesian estimation, since the distribution of $\nu_c \Lambda_g^{(c)}$ is
  completely unknown, so we resort to
  \[ \frac{1}{m} \sum_{c \in C} \tau (\nu_c \Lambda_g^{(c)}) \approx \tau
     \left( \psi \left( \frac{1}{m} \sum_{c \in C} \varphi (\nu_c
     \Lambda_g^{(c)}) \right) \right) = \tau (\psi (\tilde{X}_g)) \approx \tau
     (\psi (X_{g, \ast})) \]
  which is reasonable in the regions of linearity of $\tau \circ \psi$, so
  both in the approximate range $[0, 1]$ and above about $4$, which are most
  common for scRNA-seq data.
  
  Finally we get the estimate,
  \begin{multline*}
    \lambda_g = \lambda_g \nu_{\ast} \approx \frac{1}{m} \sum_{c \in
      C} \nu_c     \Lambda_g^{(c)} \\
    \approx \psi (X_{g, \ast}) + \frac{1}{2} \psi'' (X_{g,
     \ast}) \cdot \left[ \frac{m - 1}{m} S_{g, \ast}^2 - \psi (X_{g, \ast}) +
     X_{g, \ast}^2 \right] \indef \check{\lambda}_g . 
 \end{multline*}
 The method for $\check{\nu}_c$ is analogous. In fact, one could reproduce
  the same passages to get
  \[ \nu_c \lambda_{\ast} \approx \frac{1}{n} \sum_{g \in G} \nu_c
     \Lambda_g^{(c)} \approx \ldots = \tilde{\nu}_c \]
  by which
  \[ \check{\nu}_c \approx \frac{\nu_c \lambda_{\ast}}{\nu_{\ast}
     \lambda_{\ast}} = \nu_c . \]
  The result for $\check{\mu}_{g, c}$ follows.
\end{proof}

\section{Probability of zero reads}

\label{s:probzero}In this section we want to build on the estimates of
$\mu_{g, c}$ introduced in Section~\ref{s:estimation} to get an estimate of
the probability that $R_{g, c} = 0$. In what follows the symbol
$\tilde{\mu}_{g, c}$ denotes either the average or the square-root estimator
of $\mu_{g, c}$.

\begin{proposition}
  The probability of zero read counts can be expressed as
  \[ P (R_{g, c} = 0) = e^{- \eta_g (\mu_{g, c})} \]
  where
  \[ \eta_g (x) \defin - \log E [e^{- x \Lambda_g^{(c)} / \lambda_g}] = -
     \log \int e^{- xt / \lambda_g} \mathd \mathcal{L}_g (t) . \]
\end{proposition}

Here $\eta_g$ is the log-mgf of the law $\mathcal{L}_g$ of $\Lambda_g^{(c)}$
(which does not depend on $c$) rescaled by its mean $\lambda_g$.

\begin{proof}
  By conditioning on $\Lambda_g^{(c)}$, and using the conditional Poisson
  distribution of $R_{g, c}$, we get,
  \[ P (R_{g, c} = 0) = E [E [\mathbbm{1}_{R_{g, c} = 0} | \Lambda_g^{(c)}]] =
     E [e^{- \nu_c \Lambda_g^{(c)}}] \indef e^{- \eta_g (\nu_c
     \lambda_g)} = e^{- \eta_g (\mu_{g, c})} . \qedhere \]
\end{proof}

In full generality we cannot determine the functions $\eta_g$ for the
different genes, because the distributions $\mathcal{L}_g$'s are unknown;
nevertheless some properties of log-mgfs are universal, in particular $\eta_g$
starts from the origin, is monotone increasing, concave and has derivative $1$
in 0.

Instead of trying to estimate $\eta_g (x)$, we choose to model it with a
universal one-parameter family $(f_a)_{a \in \mathbb{R}} $ of functions $f_a :
\mathbb{R}_+ \rightarrow \mathbb{R}_+$ with the same properties: $f_a (0) =
0$, $f_a' (0) = 1$, $f_a$ monotone increasing and concave.

A simple, natural choice, based on $\log (1 + x)$, is $\frac{1}{a} \log (1 +
ax)$ for $a > 0$, which we choose to extend with continuity to
\begin{equation}
  \label{e:def-fa} f_a (x) \defin \left\{\begin{array}{ll}
    \frac{1}{a} \log (1 + ax), & a > 0\\
    (1 - a) x & a \leq 0.
  \end{array}\right.
\end{equation}
\begin{remark}
  This model, for $a > 0$, corresponds to the gamma distribution with shape
  parameter $a^{- 1}$, so we are implicitly making the assumption that
  (dropping the dependence on $g$) $\Lambda \sim \text{gamma} (a^{- 1}, a
  \lambda)$, so that $E (\Lambda) = \lambda$, $\operatorname{Var} (\Lambda) = a
  \lambda^2$, and that the read counts $R$ are negative binomial with $E (R) =
  \nu \lambda \indef \mu$ and $\operatorname{Var} (R) = \mu + a \mu^2$, see
  equation~{\eqref{e:Var-R}}.
\end{remark}

We would like to use this model to infer, for any gene $g$, some value $a (g)
\geq 0$, to which there corrisponds a reasonable estimate of the
probability of zero reads in a cell $c$, and to do so we impose the condition
that the marginal \emph{expected} number of zeros for gene~$g$ equals the
marginal \emph{observed} number of zeros:
\begin{equation}
  \label{e:a-condition} \sum_{c \in C} e^{- f_{a (g)} (\tilde{\mu}_{g, c})} =
  \sum_{c \in C} \mathbbm{1} (R_{g, c} = 0),
\end{equation}
and solve this equation for $a (g)$. We remark that here $\tilde{\mu}_{g, c}$
denotes either the average or the square-root estimator of $\mu_{g, c}$.

\begin{definition}
  \label{d:rho}We call {\emph{chance of expression}} of gene $g$ in cell $c$,
  the quantity $\rho_{g, c} \defin 1 - e^{- f_{a (g)} (\tilde{\mu}_{g, c})}$
  with $a (g)$ which solves condition~{\eqref{e:a-condition}}.
\end{definition}

The following remarks and proposition will clarify that this is both a good
definition and a reasonable one and that $\rho_{g, c} \approx P (R_{g, c}
\geq 1)$.

\begin{remark}
  Condition~{\eqref{e:a-condition}} is both necessary and
  natural. Necessary because for arbitrary $a$ the quantities on left
  and right-hand side are typically very different, and would yield
  false positives in the tests we will be performing. Natural because
  it is completely analogous to what is done for classical contingency
  tables, where the unknown probabilities of the categories are
  estimated from the proportions of the observed marginals, in such a
  way that the analogous of equation~{\eqref{e:a-condition}} holds.
\end{remark}

\begin{remark}
  We had to extend the definition of equation~{\eqref{e:def-fa}} to the
  negative values of $a$ in order to be able to solve
  equation~{\eqref{e:a-condition}} for all samples. In fact since
  $\tilde{\mu}_{g, c}$ and $R_{g, c}$ are both random variables, it may well
  be that for some genes $g$,
  \[ \sum_{c \in C} e^{- \tilde{\mu}_{g, c}} > \sum_{c \in C}
    \mathbbm{1} (R_{g, c} = 0), \] and in that case no positive value
  of $a(g)$ will satisfy condition {\eqref{e:a-condition}}. (In fact,
  in our synthetic datasets this happened for 5\% to 30\% of the
  genes. See Section~\ref{s:synthetic}.)

  The choice of $(1 - a) x$ is a simple, natural family of maps that
  extend with continuity the definition given for $a > 0$ to the
  ``forbidden'' region of the plane.  The interpretation is that in
  these cases $\tilde{\mu}_{g, c}$ could be underestimating
  $\mu_{g, c}$ and hence
  $f_{a (g)} (\tilde{\mu}_{g, c}) = (1 - a) \tilde{\mu}_{g, c} >
  \tilde{\mu}_{g, c}$ may correct the error in a suitable way.
\end{remark}

The following statement shows that $a (g)$ can be computed numerically with
ease, for example by bisection, for each gene $g \in G$.

\begin{proposition}
  The value $a (g)$ such that condition~{\eqref{e:a-condition}} holds, is
  always uniquely determined as long as $\sum_{c \in C} R_{g, c} > 0$.
\end{proposition}

\begin{proof}
  We will prove that the map $\tau : \mathbb{R} \rightarrow \mathbb{R}_+$
  defined by
  \[
    \tau (a) \defin \sum_{c \in C} e^{- f_a (\tilde{\mu}_{g, c})}
  \]
  is a bijection under the hypothesis.
  
  A direct computation shows that $f_a (x)$ is monotone decreasing in $a$ for
  all $x > 0$. In fact,
  \[ \partial_a f_a (x) = \left\{\begin{array}{ll}
       \frac{1}{a^2} \left[  \frac{ax}{1 + ax} - \log (1 + ax) \right] & a >
       0\\
       - x & a < 0
     \end{array}\right. \]
  and above formula is always negative, since
  \[ - \log (1 + ax) = \log \left( 1 - \frac{ax}{1 + ax} \right) < -
     \frac{ax}{1 + ax}, \]
  moreover $f_a (x)$ is continuous in $a$ for $a = 0$. We deduce immediately
  that $\tau$ is monotone increasing as long as $\tilde{\mu}_{g, c} > 0$ for
  some $c \in C$. The condition is true both for average and for square-root
  estimators if $R_{g, c} > 0$ for some $c \in C$. Finally it is trivial to
  observe that $\lim_{a \rightarrow - \infty} \tau (a) = 0$ and $\lim_{a
  \rightarrow + \infty} \tau (a) = + \infty$, so $\tau$ is a bijection.
\end{proof}

\section{Co-expression tables}

\label{s:tables}In this section we present a completely new tool for
measuring and testing the co-expression of two genes, and introduce
two useful statistical methods which considerably extend its scope.

Co-expression is a meaningful concept when the population of cells is
not completely homogeneous, because in that case each gene is assumed
to be independently expressed in all the cells of the sample (so that
$Q$ is supposed to be diagonal, see Section~\ref{s:model}), and hence
the read counts $R_{g,c}$ are all independent random variables.

In the case of non-homogeneous population, we assume that different
cell types can be found in the sample, each type with different genes
expressed, and hence two genes could have positive read counts in the
same cells more (or less) often that should be expected if the
population was homogeneous. Therefore genes co-expression can be a
powerful yet indirect tool to infer cell type
profiles~\cite{coex2018}.

Our approach to assess co-expression builds on the assumption that
cell differentiation will typically shun to zero the expression of
several genes and that most genes have so low expression at the single
cell level that measuring fold change is not very informative.

Based on this assumption, our main test compares the number of cells with zero
read count in couples of genes (jointly versus marginally), in a way similar
to $2 \times 2$ contingency tables, but generalized to experimental units with
different efficiency.

\begin{definition}
  For any pair of genes $g_1, g_2 \in G$, their \emph{co-expression
    table} is a contingency table of the form
  \[ \begin{array}{cc|c}
       O_{1, 1} & O_{1, 0} & O_{1, \kreuz}\\
       O_{0, 1} & O_{0, 0} & O_{0, \kreuz}\\
       \hline
       O_{\kreuz, 1} & O_{\kreuz, 0} & m
     \end{array} \]
  where $O_{1, 1}$ is the number of cells with non-zero read count for both
  genes, $O_{1, 0}$ is the number of cells with non-zero read count for $g_1$
  and zero read count for $g_2$ and so on,
  \begin{equation}
    \label{e:def-Oij} O_{i, j} \defin \# \{ c \in C \text{such that} i
    =\mathbbm{1} (R_{g_1, c} \geq 1) \text{and} j =\mathbbm{1} (R_{g_2, c}
    \geq 1) \}
  \end{equation}
  and where the marginals are as usual the sums of rows and columns and we
  recall that $m =\#C$ is the total number of cells,
  \begin{gather*}
    O_{i, \kreuz} \defin  O_{i, 1} + O_{i, 0}
    ,\qquad
    O_{\kreuz, j} \defin  O_{1, j} + O_{0, j}\\
    m =  O_{\kreuz, 1} + O_{\kreuz, 0} = O_{1, \kreuz} + O_{0, \kreuz} .
  \end{gather*}
\end{definition}

\begin{remark}
  We stress that all the information on how large is $R_{g, c}$ is
  ignored. We consider only the two cases $R_{g, c} = 0$ and
  $R_{g, c} \geq 1$. In principle this may be a weakness of this
  approach, but one should recall that very few genes have high
  counts, and this method is particularly suited to deal with low
  espressions and small integer counts which are typical in scRNA-seq
  databases.
\end{remark}

\subsection{Classical contingency tables}

The naive approach with classical contingency tables does not work for our
proposed scRNA-seq model, because the variability of efficiency $\nu_c$
between cells creates spurious correlation.

Consider for example the co-expression table below, relative to two
\emph{constitutive} genes (which, as such, should be expressed in
\emph{all} cells),
\[ \begin{array}{cc|c}
     705 & 4 & 709\\
     654 & 16 & 670\\
     \hline
     1359 & 20 & 1379
   \end{array} \]
The marginals of gene $g_1$ are $O_{1, \kreuz} = 709$ and $O_{0, \kreuz} =
670$, with a ratio $\frac{O_{0, \kreuz}}{m} = \frac{670}{1379} \approx
\frac{1}{2}$, showing that in 1 cell out of 2 there are zero read counts for this
gene. Despite the fact that gene $g_1$ should be certainly expressed in all
cells, this can be explained because of the combination of low biological
expression and low extraction efficiency. Something analogous happens for gene
$g_2$, with a corresponding ratio of about $\frac{1}{70}$.

These ratios suggest that any cell has a probability of $\frac{1}{2}$ of
having zero read count of $g_1$ and a probability of $\frac{1}{70}$ of having
zero read count of $g_2$. These two events would be independent if all cells
had the same extraction efficiency: together with the independence of RNA
fragments extraction, this would yield the expected read counts of
classical contingency tables; for example $1379
\cdot \frac{1}{2} \cdot \frac{1}{70} \approx 9.7$ and similarly,
\[ \begin{array}{cc|c}
     698.7 & 10.3 & 709\\
     660.3 & 9.7 & 670\\
     \hline
     1359 & 20 & 1379
   \end{array} \]
 Since $4$ is quite far from $10.3$, giving alone a $2 \sigma$
 deviation from the null hypothesis, the classical contingency table analysis
 would give high significance to the false hypothesis that the two constitutive
 genes are positively co-expressed, suggesting that there are at least two different
 categories of cells: cells in which both are expressed and cells where neither
 is expressed.

What's really happening is that there are cells with high efficiency and cells
with low efficiency. While the matematical model of contingency tables builds
on the assumption that all experimental units are identically distributed,
this does not hold in the case of scRNA-seq data.

\subsection{Expected counts in co-expression tables}

The definition of the {\emph{observed}} cells $O_{i, j}$ given by
equation~{\eqref{e:def-Oij}} can be rewritten more succintly as
\begin{equation}
  \label{e:def-Oij-2} O_{i, j} = \sum_{c \in C} \mathbbm{1} (R_{g_1, c}
  \geq 1)^i \cdot \mathbbm{1} (R_{g_1, c} = 0)^{1 - i} \cdot \mathbbm{1}
  (R_{g_2, c} \geq 1)^j \cdot \mathbbm{1} (R_{g_2, c} = 0)^{1 -
    j}.
\end{equation}
for $i,j\in\{0,1\}$.

Since we are going to build a statistical test for the independence of
expression of the two genes, we put ourselves in the null hypothesis, and
deduce that the {\emph{expected}} number of cells under independence $\epsilon_{i, j}
\defin E_{H_0} (O_{i, j})$ is
\[ \epsilon_{i, j} = \sum_{c \in C} P (R_{g_1, c} \geq 1)^i P (R_{g_1, c} =
   0)^{1 - i} P (R_{g_2, c} \geq 1)^j P (R_{g_2, c} = 0)^{1 - j} . \]
By Definition~\ref{d:rho}, $\epsilon_{i, j}$ can be estimated using the chance of
expression of the genes in each cell.

\begin{definition}
  For any pair of genes $g_1, g_2 \in G$, their {\emph{table of expected cell
  counts under the hypothesis of independence}} (``table of expected'' for
  short) is given by
  \begin{equation}
    \label{e:def-Etij} \tilde{\epsilon}_{i, j} \defin \sum_{c \in C} \rho_{g_1, c}^i
    (1 - \rho_{g_1, c})^{1 - i} \rho_{g_2, c}^j  (1 - \rho_{g_2, c})^{1 - j},
    \qquad i, j \in \{ 0, 1 \}
  \end{equation}
  where $\rho_{g, c}$ denotes the chance of expression of gene $g$ in cell
  $c$.
\end{definition}

The values $O_{i, j}$ and $\tilde{\epsilon}_{i, j}$ for $i, j \in \{ 0, 1 \}$ given
by equations~{\eqref{e:def-Oij}} or~{\eqref{e:def-Oij-2}}, and
equation~{\eqref{e:def-Etij}} define the tables of observed and expected
cells, with the property that marginals are the same, thanks to
condition~{\eqref{e:a-condition}}.

For example, for the experiment with the two constitutive genes presented
above, the table with the expected cells (average estimators) is the
following,
\[ \begin{array}{cc|c}
     703.6 & 5.4 & 709\\
     655.4 & 14.6 & 670\\
     \hline
     1359 & 20 & 1379
   \end{array} \]
As can be seen, these values are much closer to the observed.

\subsection{Co-expression estimator and test}

The interpretation of the co-expression table is now performed in a
way similar to usual contingency tables, with a test for independence
of expression based on the $\chi^2 (1)$ distribution, and with an
additional \emph{co-expression index}, based on the same framework and
similar in principle to a classical correlation.

We stress that this is an approximate test, and one cannot prove in
full generality that the $\chi^2(1)$ distribution is exactly correct
for the statistics under the null hypothesis. However we used the
synthetic datasets described in Section~\ref{s:synthetic} to get the
empirical distribution of the $\chi^2(1)$ $p$-value and found it in
good accordance with the theory, which prescribes uniform
distribution. Figure~\ref{f:coex_test_fp} shows the result.

On the other hand the statistical power of this test emerges from the
direct application to real data, which will be detailed in the twin
paper.

\begin{definition}
  Given two genes with co-expression table $(O_{i, j})_{i, j = 0, 1}$ and
  table of expected $(\tilde{\epsilon}_{i, j})_{i, j = 0, 1}$, the statistics of the
  test for the independence of expression is
  \[ W \defin \sum_{i, j = 0}^1 \frac{(O_{i, j} - \tilde{\epsilon}_{i, j})^2}{1 \vee
     \tilde{\epsilon}_{i, j}} . \]
  The co-expression index is
  \[ R \defin \biggl( \sum_{i, j = 0}^1 \frac{1}{1 \vee \tilde{\epsilon}_{i, j}}
     \biggr)^{- 1 / 2} \cdot \sum_{i, j = 0}^1 (- 1)^{i + j}  \frac{O_{i, j} -
     \tilde{\epsilon}_{i, j}}{1 \vee \tilde{\epsilon}_{i, j}} \].
\end{definition}

The statistics $W$ is defined in analogy with the traditional
contingency tables, with a regularizing correction at the denominator
to take into account the fact that $\tilde{\epsilon}_{i, j}$ could be
much smaller than 1 for some low expression genes. In this way those
cases will not become false positives.

The co-expression index $R$ is defined in such a way that
$| R | = \sqrt{W}$ with the sign that encodes the direction of the
deviation from independence, so it will be positive when the genes are
positively co-expressed, negative in the opposite case, and
$\mathcal N(0,1)$-distributed when there is independence.

\begin{figure}[ht]
  \includegraphics[width=\textwidth]{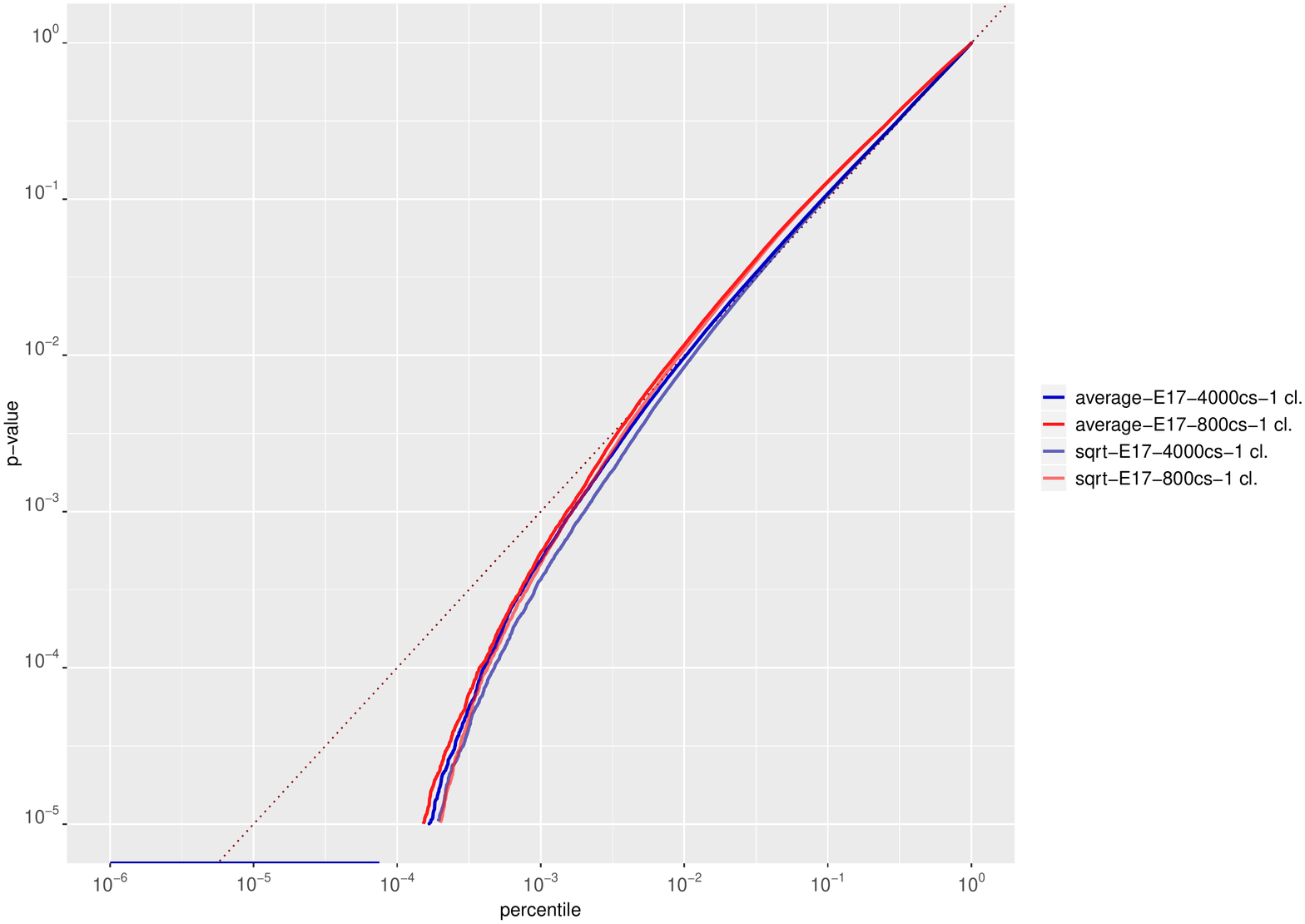}
  \caption{Empirical distribution of $p$-value computed with
    $\chi^2(1)$ quantiles.  To be under the null hypothesis, we used
    the four 1-cluster synthetic datasets (see
    Section~\ref{s:synthetic}); we performed co-expression tests for
    all pairs of genes, and then randomly sampled a subset of $10^6$
    points to plot the lines. The plot proves that these tests have
    the correct incidence of false positive for all significance
    values greater than about $0.005$.}
  \label{f:coex_test_fp}
\end{figure}

\begin{proposition}
  Given two genes with co-expression table $(O_{i, j})_{i, j = 0, 1}$ and
  table of expected $(\tilde{\epsilon}_{i, j})_{i, j = 0, 1}$, define for $i, j \in
  \{ 0, 1 \}$,
  \[
    Z_{i, j} \defin \frac{O_{i, j} - \tilde{\epsilon}_{i, j}}{\sqrt{1
        \vee \tilde{\epsilon}_{i, j}}}
    \qquad\text{and}\qquad
    v_{i, j} \defin \frac{(- 1)^{i + j}}{\sqrt{1
        \vee \tilde{\epsilon}_{i, j}}}
  \]
and let $W$ and $R$ be as above. Then $R =
  \frac{v}{\| v \|} \cdot Z$, and $W = \| Z \|^2$ by definition and in the
  vector space $\mathbb{R}^4$, $Z$ and $v$ have the same direction, so $R^2 =
  W$.
  
  If the components of $Z$ are supposed to be standard Gaussian, independent
  but conditioned on the values of the marginals of the tables, then $R$ is
  standard Gaussian and $W$ is a chi-square with 1 degree of freedom.
\end{proposition}

\begin{proof}
  Firstly notice that, given the marginals, the value of any cell determines
  the other three, and the following relations hold:
  \[
       O_{0, 0} \lessgtr \tilde{\epsilon}_{0, 0} \Leftrightarrow O_{1, 1} \lessgtr
       \tilde{\epsilon}_{1, 1} \Leftrightarrow O_{0, 1} \gtrless \tilde{\epsilon}_{0, 1}
       \Leftrightarrow O_{1, 0} \gtrless \tilde{\epsilon}_{1, 0}
  \]
  hence $O_{i, j} = \tilde{\epsilon}_{i, j} + (- 1)^{i + j} r$ for some suitable $r
  \in \mathbb{R}$ not depending on $i$ and $j$. Then $Z_{i, j} = rv_{i, j}$
  and the two vectors have the same direction.
  
  For the second part of the statement, conditioning on the values of
  the marginals is equivalent to restricting $Z$ to the 1-dimensional
  subspace $\operatorname{Span} (v)$. Since the covariance matrix of
  $Z$ before conditioning is the identity, it is invariant by
  rotations and therefore $R$, which is the projection on the
  subspace, has standard Gaussian distribution, and finally
  $W = R^2 \sim \chi^2 (1)$.
\end{proof}

\begin{corollary}
  Under the null hypothesis of the test for independence of expression,
  \[ R \dot{\sim} \mathcal{N} (0, 1) \qquad \text{and}
     \qquad W \dot{\sim} \chi^2 (1) . \]
\end{corollary}

\begin{proof}
  Under $H_0$ we expect
  $E (O_{i, j}) \approx \tilde{\epsilon}_{i, j}$, so the components of
  $Z$ are approximately standard Gaussian, and they are independent
  before conditioning to the marginals.
\end{proof}

\subsection{Extensions}

The framework of co-expression tables allows the introduction of some
additional tools.

\subsubsection{Differential expression analysis}

When the cells $C$ are divided into $k \geq 2$ different groups, $C =
\bigcup_{j = 1}^k C_j$ (called {\emph{conditions}}), it is important to
verify, for each gene $g$, if there is a significant difference of expression
between the groups. This is a very active research field of its own,
as can be see for example in~\cite{Deseq2-2014,SCDE2014,MAST2015,decent2019}
and references therein.

In our framework this test can be done by means of an
expression/condition table, similar to the co-expression tables of the
main result, but with as first variable the gene $g$ (collapsed in
categories $\{ R_{g, c} \geq 1 \}$ and $\{ R_{g, c} = 0 \}$) and
as second variable the condition. Formally, one can define
\[ O_{i, j} \defin \# \{ c \in C_j \text{ such that } i =\mathbbm{1}
   (R_{g_{}, c} \geq 1) \}, \qquad i = 0, 1,
   \quad j = 1, 2, \ldots, k \]
and estimate the expected cell counts under the hypothesis of independence
with
\[ \tilde{\epsilon}_{i, j} \defin \sum_{c \in C_j} \rho_{g, c}^i  (1 - \rho_{g,
   c})^{1 - i}, \qquad i, j \in \{ 0, 1 \} . \]
Then the test goes on as in a classical contingency table, by the
approximation that under the null hypothesis,
\[ W \defin \sum_{i = 0}^1 \sum_{j = 1}^k \frac{(O_{i, j} - \tilde{\epsilon}_{i,
   j})^2}{1 \vee \tilde{\epsilon}_{i, j}} \dot{\sim} \chi^2 (k - 1) . \]

\subsubsection{Global differentiation index}

When the co-expression index is computed genome-wide, that is, for all pairs
of genes $(g_1, g_2) \in G \times G$, it makes possible to score the genes by
global differentiation inside the sample. This is another important
field of research, see for example~\cite{BreAnd2013,SCDD2016}.

Several different statistics may be proposed, and we found the
following to be relevant and informative.

\begin{definition}
  The {\emph{global differentiation index}} (GDI) for a gene $g\in G$
  is the quantity
  \[
    \text{GDI}(g)\defin\log(-\log(1-F_{\chi^2(1)}(S_g))),
  \]
  where $F_{\chi^2(1)}$ is the $\chi^2(1)$ cumulative distribution
  function, $\log$ denotes the natural logarithm, and $S_g$ is a very
  high percentile for the test statistics,
  \[
    S_g \defin P_{1 - \alpha} \{ R^2_{g, h} : h \in G \}, \quad g \in G.
  \]
  Here $R_{g, h}$ denotes the co-expression index between $g$ and $h$,
  $P_x (A)$ denotes the $x$-percentile of the sample $A$, and we
  typically set $\alpha = 10^{- 3}$ for a genome $G$ of about $15000$
  genes.
\end{definition}

\begin{figure}[p]
  \label{f:GDI}
  \includegraphics[width=0.8\textwidth]{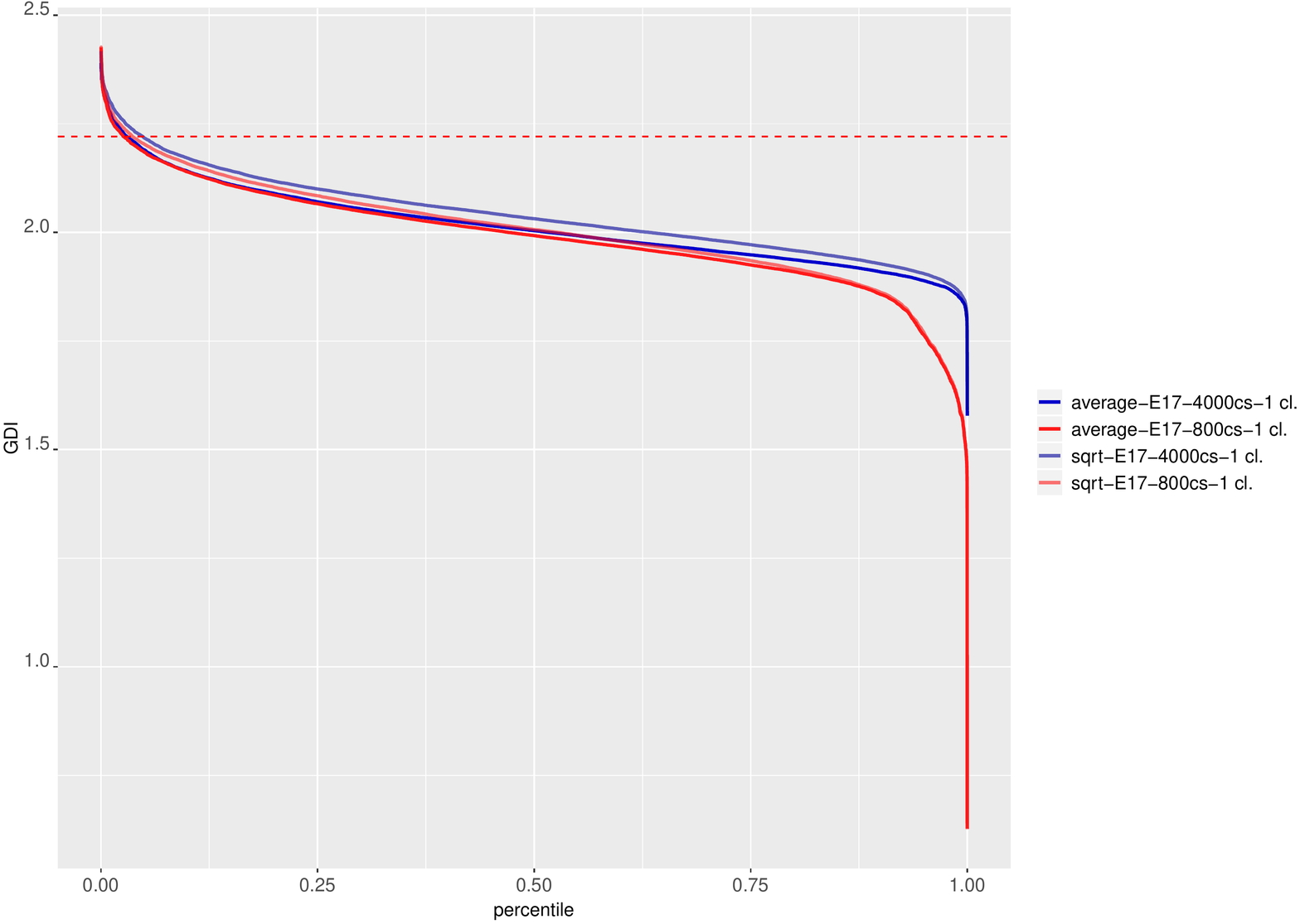}

  \medskip

  \includegraphics[width=0.8\textwidth]{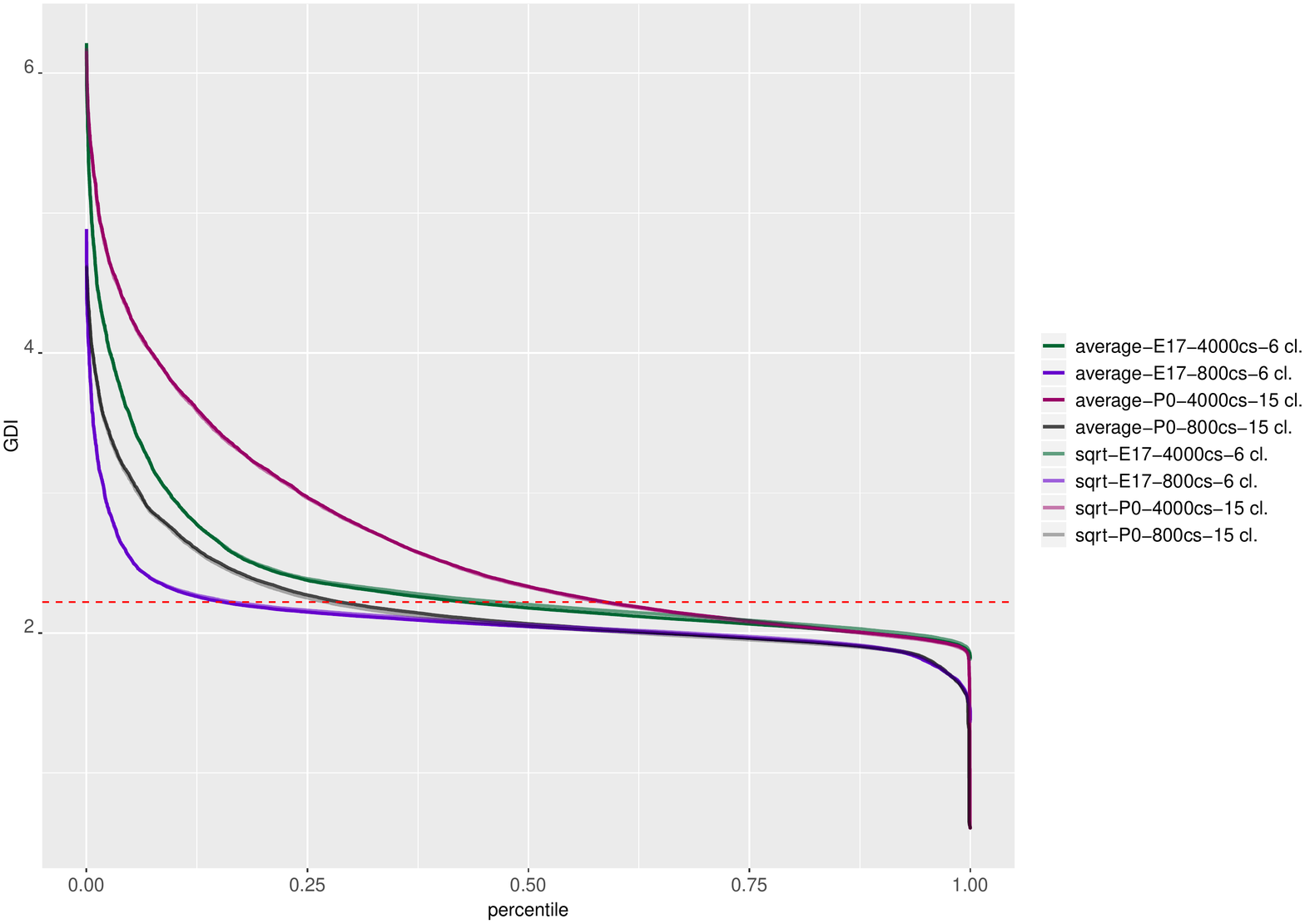}
  \caption{Empirical distribution of GDI from the synthetic datasets
    (see Section~\ref{s:synthetic}).  The first plot is under the null
    hypothesis, as we used the four 1-cluster datasets; the second
    plot is under the alternative hypothesis for a unknown but large
    fraction of the genes, as we used the eight multiple-cluster
    datasets.  The dotted line corresponds to the $10^{-4}$ quantile
    for the approximated global differentiation test. The threshold is
    about $2.2203$ on the GDI scale and about $15.137$ on the $S_g$
    scale. Under the null hypothesis the false positive were
    between 3\% and 5\%.}
\end{figure}

Although the distribution of $S_g$ is difficult, this index (or
$\text{GDI}(g)$, which is just a convenient rescaling of $S_g$) can be
qualitatively used to score genes by how much they are differentiated,
and even to design an approximated test of global differentiation:
from verification with synthetic datasets, under the null hypothesis
that the gene is not differentiated, we found that
$P_{H_0} (S_g > F_{\chi^2(1)}(1-10^{- 4}))$ was between 3\% and 5\%,
so it is possible to use the approximate quantile $10^{- 4}$ for this
statistics.

\section{Synthetic datasets}

\label{s:synthetic}Since several of the conclusion in this work are of
approximate nature, we used Monte Carlo simulation to test their validity, by
generating some synthetic datasets.

Since much depends on the realism of the generated data, we took two
real scRNA-seq datasets, labelled P0 and E17, with different
extraction techniques and very different size and typical extraction
efficiency; we clustered the cells with standard techniques finding 15
clusters for P0 and 6 for E17; inside each cluster separately we
performed a maximum likelihood estimation of all the parameters (the
extraction efficiency $\nu_c$ for all cells, and the two parameters of
the gamma distribution for $\Lambda_g$ for all genes and all
clusters).

For both set of parameters estimated from P0 and E17, we generated 4
random datasets, with different number of cells (800 and 4000) and
both in differentiated and indifferentiated conditions, that is,
either with all clusters or with all cells sampled from just one
cluster.

All datasets were analyzed with our framework, both with average estimators
and with square-root estimators. The value of the parameters was compared with
their estimates, and the distribution of $R$, $W$ and $S_g$ was controlled in
the indifferentiated condition.

\paragraph{Acknowledgements.}
Both author are thankful to Marco Pietrosanto, Manuela
Helmer-Citterich and Federico Cremisi for the precious support and
enlightening discussions.

\bibliographystyle{unsrt}
\bibliography{cotan}

\end{document}